\documentclass[12pt,letterpaper]{article}
\usepackage[T1]{fontenc}
\usepackage[utf8]{inputenc}
\usepackage[english]{babel}

\usepackage{microtype}
\usepackage{amsmath,amssymb,amsthm,mathtools}
\usepackage{newtxmath}
\usepackage{bm}
\usepackage{dsfont}
\usepackage{setspace}

\usepackage[
  letterpaper,
  top=0.95in,
  bottom=1.00in,
  left=1.25in,
  right=1.25in,
  footskip=0.42in
]{geometry}

\usepackage{graphicx}
\usepackage{booktabs}
\usepackage{threeparttable}
\usepackage{array}
\usepackage{siunitx}
\usepackage{caption}
\usepackage{subcaption}

\usepackage{enumitem}
\setlist{nosep,leftmargin=1.5em}
\usepackage{titlesec}
\titleformat{\section}[block]
  {\centering\large\bfseries}
  {\thesection.}{0.65em}{}
\titleformat{\subsection}
  {\normalsize\bfseries}
  {\thesubsection}{0.65em}{}
\titleformat{\subsubsection}[runin]
  {\normalsize\itshape}
  {\thesubsubsection}{0.65em}{}[.]
\titlespacing*{\section}
  {0pt}{2.7ex plus .6ex minus .2ex}{1.1ex}
\titlespacing*{\subsection}
  {0pt}{2.1ex plus .4ex minus .2ex}{0.7ex}
\titlespacing*{\subsubsection}
  {0pt}{1.7ex plus .3ex minus .2ex}{0.8em}

\usepackage{titling}
\pretitle{\begin{center}\LARGE}
\posttitle{\par\end{center}\vspace{-0.75em}}
\preauthor{\begin{center}\large}
\postauthor{\par\end{center}\vspace{-1.33em}}
\predate{\begin{center}\large}
\postdate{\par\end{center}\vspace{0.17em}}
\renewenvironment{abstract}{  \begin{center}\bfseries Abstract\end{center}  \vspace{-0.45em}  \par\noindent\small
}{  \par
}

\usepackage[round,authoryear]{natbib}
\usepackage{indentfirst}
\usepackage[bottom,hang,flushmargin]{footmisc}
\usepackage{xcolor}
\definecolor{ink}{HTML}{000000}
\definecolor{linkblue}{HTML}{315A75}
\color{ink}

\usepackage[
  colorlinks=true,
  hyperfootnotes=false,
  linkcolor=linkblue,
  citecolor=linkblue,
  urlcolor=linkblue,
  filecolor=linkblue,
  pdfborder={0 0 0}
]{hyperref}
\usepackage[capitalize,noabbrev]{cleveref}

\newtheoremstyle{paperstyle}
  {8pt}{8pt}{\itshape}{}{\bfseries}{.}{0.5em}{}
\theoremstyle{paperstyle}
\newtheorem{proposition}{Proposition}[section]

\theoremstyle{definition}
\newtheorem{assumption}{Assumption}

\crefname{assumption}{assumption}{assumptions}
\Crefname{assumption}{Assumption}{Assumptions}
\crefname{proposition}{proposition}{propositions}
\Crefname{proposition}{Proposition}{Propositions}
\crefname{remark}{remark}{remarks}
\Crefname{remark}{Remark}{Remarks}
\crefname{corollary}{corollary}{corollaries}
\Crefname{corollary}{Corollary}{Corollaries}

\newcommand{\E}{\mathbb{E}}

\newcommand{\Var}{\operatorname{Var}}

\newcommand{\ind}{\mathds{1}}

\newcommand{\SimMaxBias}{0.011}
\newcommand{\SimCoverageMin}{0.940}
\newcommand{\SimCoverageMax}{0.968}
\newcommand{\StdMaxBias}{0.011}
\newcommand{\MhrsCountPre}{2,963}
\newcommand{\MhrsCountPost}{2,874}

\newcommand{\ComparisonStateCount}{47}

\newcommand{\AggregateATT}{4.64}
\newcommand{\AggregateLower}{0.75}
\newcommand{\AggregateUpper}{8.44}
\newcommand{\HSLessRelative}{3.20}
\newcommand{\SomeCollegeRelative}{-0.77}
\newcommand{\BAPlusRelative}{-2.45}
\newcommand{\HSLessCombined}{7.85}
\newcommand{\SomeCollegeCombined}{3.87}
\newcommand{\BAPlusCombined}{2.20}
\newcommand{\HSLessStandardizedChange}{10.15}
\newcommand{\TargetAverageChange}{6.95}
\newcommand{\LowHighContrast}{5.65}
\newcommand{\LowHighLower}{2.78}
\newcommand{\LowHighUpper}{8.47}
\newcommand{\EducationOmnibus}{15.58}

\newcommand{\BRFSSLow}{6.56}
\newcommand{\BRFSSMiddle}{2.51}
\newcommand{\BRFSSHigh}{2.73}

\newcommand{\OASection}[1]{Online Appendix Section~\ref{#1}}

\title{Heterogeneous Policy Effects in Comparative Case Studies with Treated-Unit Microdata
}

\author{  Carl Bonander\thanks{    Karlstad Business School and the Center for Societal Risk Research,
    Karlstad University, Karlstad, Sweden.
    Correspondence: \href{mailto:carl.bonander@kau.se}{carl.bonander@kau.se}.
  }
}

\date{October 14, 2026}

\begin{document}

\pagenumbering{gobble}

\maketitle
\thispagestyle{empty}

\begin{center}
\begin{minipage}{0.80\textwidth}

\begin{abstract}

This article develops an identification framework for heterogeneous policy effects when individual-level data are available only in the treated jurisdiction. The framework combines outcome changes from repeated cross-sections within that jurisdiction with a population-average effect identified from aggregate panel data. Under conditional relative parallel trends and overlap, this information identifies subgroup effects standardized to the treated jurisdiction's overall post-treatment covariate distribution, but generally not effects averaged over each subgroup's own distribution. I establish how these subgroup-average effects can nevertheless be recovered using covariate-specific aggregate effects or additional pretreatment observations under assumptions about the evolution of untreated trends. Simulations and an application to the Massachusetts health reform illustrate the framework.

\end{abstract}

\vspace{0.65em}


\end{minipage}
\end{center}

\clearpage

\pagenumbering{arabic}
\singlespacing

\section{Introduction}

Major policy reforms can prompt dedicated data collection in the implementing jurisdiction. Following the 2006 Massachusetts health reform, for example, the Massachusetts Health Reform Survey collected detailed information from nonelderly adults on insurance coverage, access to care, and affordability \citep{long2008road}. Evaluation of the Stockholm congestion-charging trial likewise drew on detailed travel surveys and administrative data \citep{eliasson2009stockholm}, while evaluation of minimum unit pricing in Scotland included repeated surveys alongside administrative and other population data \citep{katikireddi2019mup}. Such data can also contain income, age, and other characteristics relevant to treatment-effect heterogeneity. When potential comparison jurisdictions are observed only in aggregate, these characteristics may be unavailable or measured much more coarsely outside the jurisdiction implementing the policy.

Systematic variation in policy effects can matter for targeting and subsequent policy design \citep{manski2004statistical,athey2021policy}. Many methods for estimating subgroup effect levels require comparable individual-level information on subgroup characteristics in treated and comparison populations. Recent difference-in-differences methods, for example, use repeated cross-sections from both populations to estimate heterogeneous effects \citep{deb2024flexible,wooldridge2026nonlinear}. Related methods use richer individual-level outcome information from treated and comparison units: changes-in-changes uses individual-level outcome distributions in treated and comparison populations \citep{athey2006identification}, while extensions of synthetic controls use individual observations \citep[e.g.,][]{robbins2017framework,abadie2021penalized} or outcome distributions \citep{gunsilius2023distributional} from the comparison units. A different strategy is to combine an aggregate comparative-case estimate with subgroup information available only in the treated jurisdiction. In their study of the Flint water crisis, \citet{trejo2024flint} use a synthetic-control counterfactual together with subgroup-specific changes within Flint, including comparisons by household water-service-line material, which is unavailable in the comparison districts.

In this paper, I develop an identification framework for this empirical setting, establishing which subgroup effects can be recovered from aggregate comparisons and treated-jurisdiction microdata. Within that jurisdiction, differences in pre--post outcome changes across values of an observed characteristic can identify differences in treatment effects under a relative parallel-trends assumption, as in subgroup and factorial difference-in-differences designs \citep{shahn2023subgroup,shahn2024group,xu2026factorial}. These contrasts can establish which groups experience larger effects, but not whether the policy increased or decreased the outcome within any group. Combining these contrasts with a compatible population-average effect from aggregate panel data identifies subgroup effect levels without requiring the corresponding characteristics in the comparison jurisdictions. The population-average effect can be identified using causal panel estimators that accommodate settings with a single treated unit, such as synthetic controls or related methods \citep{abadie2010synthetic,xu2017generalized,arkhangelsky2021synthetic,ben2021augmented,athey2021matrix}.

The identification problem becomes more demanding when observed covariates predict both untreated changes and treatment effects. I show that, under conditional relative parallel trends and suitable overlap, subgroup effects standardized to the target population's covariate distribution can be identified from within-jurisdiction outcome changes and a single population-average effect. The aggregate analysis need not itself be standardized, which is particularly useful when comparison jurisdictions report only overall outcomes. The same information generally does not suffice to identify effects averaged over each subgroup's own covariate distribution. I establish this nonidentification result and derive two routes to recovering subgroup-average effects using covariate-specific aggregate effects or additional pretreatment observations under assumptions about the evolution of untreated trends.

The rest of this paper proceeds as follows. Section~\ref{sec:setup} develops the identification and nonidentification results, including the roles of covariate adjustment, aggregate information, and pretreatment observations. Section~\ref{sec:estimation} discusses estimation and inference. Section~\ref{sec:numerical} examines finite-sample behavior, and Section~\ref{sec:massachusetts} applies the framework to the Massachusetts health reform and compares the results with estimates using individual-level data from the comparison states. Section~\ref{sec:discussion} concludes.

\section{Setup and identification}
\label{sec:setup}

Consider a policy adopted in one jurisdiction. Aggregate outcomes are observed for that jurisdiction and a set of untreated comparison jurisdictions, while individual-level repeated cross-sections are observed within the treated jurisdiction before and after treatment. I use \emph{aggregate comparative-case design} to refer to a design based on the aggregate panel that identifies the population-average policy effect using the treated jurisdiction and the comparison jurisdictions. The repeated cross-sections within the treated jurisdiction are used to recover how effects vary across population groups.

\subsection{A two-group illustration}
\label{subsec:simple-illustration}

To illustrate the basic argument, suppose the policy effect is of interest for two population groups, indexed by $g\in\{0,1\}$, both of which are exposed to the policy after its introduction. Let $\Delta_g$ denote the observed pre--post change and $\tau_g$ the treatment effect for group $g$. If the untreated outcome would have changed by the same amount $\delta$ in both groups, then
\[
\Delta_0=\delta+\tau_0,
\qquad
\Delta_1=\delta+\tau_1,
\qquad\text{so that}\qquad
\Delta_1-\Delta_0=\tau_1-\tau_0.
\]
The difference-in-differences comparison therefore identifies the difference between the two treatment effects, but not either effect separately.

Let $p$ denote the population share of group 1. If a comparative-case analysis identifies the population-average effect $\bar\tau=(1-p)\tau_0+p\tau_1$, and $\bar\Delta=(1-p)\Delta_0+p\Delta_1$ denotes the corresponding population-average observed change, the average and the contrast determine both group effects:
\[
\tau_g=\bar\tau+\Delta_g-\bar\Delta,
\qquad g\in\{0,1\}.
\]

\subsection{Identification from within-jurisdiction outcome changes}
\label{subsec:relative-identification}

I now formalize this argument for a general observed characteristic in repeated cross-sections. Let $t\in\{0,1\}$ index the two repeated cross-sections and let $a\in\{0,1\}$ denote policy exposure, with $a=0$ corresponding to no policy and $a=1$ to exposure to the policy. The observed outcomes satisfy $Y_0=Y_0(0)$ before treatment and $Y_1=Y_1(1)$ after treatment.\footnote{This notation incorporates consistency of the observed outcome with the realized policy state and no anticipation, so that subsequent policy adoption does not affect the pretreatment potential outcome.}

Let $Z$ denote a predetermined characteristic, or vector of characteristics, along which treatment-effect heterogeneity is of interest. Because different individuals may be observed in the two periods, let $\E_t$ denote expectation in the population represented by the period-$t$ cross-section and write $\mu_t^a(z)=\E_t[Y_t(a)\mid Z=z]$. The analysis therefore concerns changes in period-specific population means.\footnote{The same argument can be applied when individual-level panel data are available, but I focus on repeated cross-sections because the framework does not require the same individuals to be followed over time.} The post-treatment effect and the observed pre--post change are
\[
\tau(z)=\mu_1^1(z)-\mu_1^0(z),
\qquad
\Delta(z)=\mu_1^1(z)-\mu_0^0(z).
\]

Let $P_t^Z$ denote the distribution of $Z$ in the population represented by the period-$t$ cross-section. The target population is the population represented by the treated jurisdiction's post-treatment (period-$1$) cross-section, so $F_Z^\star=P_1^Z$. Its distribution can be estimated from that cross-section without a separate reference-population sample. Estimating changes for this population requires pretreatment information on the subgroups it contains. I express this overlap requirement using absolute continuity: $G\ll H$ means that $H(B)=0$ implies $G(B)=0$ for every measurable set $B$.

\begin{assumption}[Subgroup overlap]
\label{ass:subgroup-overlap}
For every survey period $t$ used in the marginal subgroup analysis,
\[
F_Z^\star\ll P_t^Z.
\]
\end{assumption}

For the two-period analysis, Assumption~\ref{ass:subgroup-overlap} applies to $t=0,1$. With categorical $Z$, it requires every subgroup with positive target probability to have positive probability in both survey populations, while allowing its share to change. With continuous $Z$, the same requirement applies to sets of characteristic values. To accommodate both cases, the propositions state identification for $F_Z^\star$-almost every $z$. This determines effects up to sets of zero target probability and includes every categorical subgroup with positive target probability. Technical conventions for conditional distributions and integrability are given in Appendix~\ref{app:proofs}.

Overlap makes the observed subgroup changes available for comparison. To interpret differences in those changes as differences in policy effects, the next assumption requires a common untreated change across subgroups.

\begin{assumption}[Relative parallel trends]
\label{ass:relative-parallel-trends}
For $F_Z^\star$-almost every $z$,
\begin{equation}
\mu_1^0(z)-\mu_0^0(z)=\delta
\label{eq:relative-parallel-trends}
\end{equation}
for an unrestricted scalar $\delta$.
\end{assumption}

Assumption~\ref{ass:relative-parallel-trends} allows untreated outcome levels to differ across subgroups but requires their changes to be equal. Because the cross-sections can represent different individuals, this requirement applies to the changes in period-specific population means, including changes in subgroup composition. All subgroups receive the policy, so comparing their observed changes then identifies differences in treatment effects.\footnote{For related results on identification of subgroup treatment-effect contrasts, see \citet{shahn2023subgroup,shahn2024group,xu2026factorial}. The analysis here examines which subgroup effect levels can be recovered when these contrasts are combined with aggregate comparative-case estimates.} Under the assumption, $\Delta(z)=\delta+\tau(z)$.

The population-average effect is
\begin{equation}
\bar\tau
=
\int\tau(z)\,dF_Z^\star(z).
\label{eq:average-effect}
\end{equation}
For this target population, $\bar\tau=\E_1[Y_1(1)-Y_1(0)]$ by iterated expectation. Write $\bar\Delta=\int\Delta(z)\,dF_Z^\star(z)$ for the corresponding average observed change. Substituting $\Delta(u)=\delta+\tau(u)$ gives
\begin{align}
\Delta(z)-\bar\Delta
&=
\{\delta+\tau(z)\}
-
\int\{\delta+\tau(u)\}\,dF_Z^\star(u)
\nonumber\\
&=
\delta+\tau(z)-\delta-\bar\tau
\nonumber\\
&=
\tau(z)-\bar\tau.
\label{eq:identified-heterogeneity}
\end{align}
Thus the within-jurisdiction comparison identifies the treatment-effect contrast $h(z)=\tau(z)-\bar\tau$, but leaves $\bar\tau$ undetermined. An aggregate comparative-case design can supply the missing population-average effect under the assumptions of the chosen panel method \citep[e.g.,][]{abadie2010synthetic,xu2017generalized,arkhangelsky2021synthetic,athey2021matrix}. I state this identification requirement generically because the following results use the identified average without depending on a particular choice of panel method.

\begin{assumption}[Aggregate identification]
\label{ass:estimand-compatibility}
The aggregate comparative-case design identifies the population-average effect $\bar\tau$ in \cref{eq:average-effect}.
\end{assumption}

Combining the identified average with the within-jurisdiction contrasts gives the following result.

\begin{proposition}[Identification of heterogeneous effects]
\label{prop:anchored-identification}
Under \cref{ass:subgroup-overlap,ass:relative-parallel-trends,ass:estimand-compatibility}, $\tau(z)$ is identified for $F_Z^\star$-almost every $z$, with
\begin{equation}
\tau(z)=\bar\tau+\Delta(z)-\bar\Delta.
\label{eq:anchored-effect}
\end{equation}
\end{proposition}

\noindent\emph{Proof.} In Appendix~\ref{app:proofs}.\par\smallskip

Proposition~\ref{prop:anchored-identification} shows that subgroup effects can be identified without observing subgroup characteristics in the comparison jurisdictions. Once the aggregate comparative-case analysis has identified the population-average effect, differences across subgroups can be identified entirely from the treated-jurisdiction repeated cross-sections. The same aggregate effect can therefore be used to examine heterogeneity across several characteristics, provided the relative parallel-trends assumption is credible for each comparison.

Using an aggregate effect in this way requires attention to the population and follow-up period it represents. For example, an aggregate estimator averaging effects over several post-treatment periods may target
\[
\bar\tau_A
=
\sum_{s\in\mathcal T_A}\omega_s\,\E_{A,s}[Y_s(1)-Y_s(0)],
\]
where the nonnegative weights sum to one and $\E_{A,s}$ averages over the treated-jurisdiction population represented by the aggregate outcome in period $s$. This average can differ from the required effect $\bar\tau=\E_1[Y_1(1)-Y_1(0)]$ when effects vary over time or the represented populations differ, even if the aggregate effect itself is causally identified.

\subsection{Covariate adjustment and standardization}
\label{subsec:covariate-adjustment}

The comparisons in Proposition~\ref{prop:anchored-identification} rely on equal untreated changes across subgroups. This assumption may be difficult to justify when, for example, age predicts changes in insurance coverage and age composition differs across education groups. Conditioning on age allows education groups to be compared among adults of the same age. If treatment effects also vary with age, however, averaging these conditional effects over each group's own age distribution answers a different question from comparing the groups at a common age distribution.

Let $X$ denote covariates unaffected by the policy, and write $\mu_t^a(z,x)=\E_t[Y_t(a)\mid Z=z,X=x]$. The conditional policy effect and observed change are
\[
\tau(z,x)=\mu_1^1(z,x)-\mu_1^0(z,x),
\qquad
\Delta(z,x)=\mu_1^1(z,x)-\mu_0^0(z,x).
\]
Let $P_t^{X,Z}$ denote the joint distribution of $(X,Z)$ in the period-$t$ population. For the period-$1$ target population, $F_{X,Z}^\star=P_1^{X,Z}$, with marginals $F_X^\star$ and $F_Z^\star$. The marginal distribution $F_X^\star$ used for standardization is estimated from this same cross-section. The subgroup-average effect, $\tau^J(z)$, averages treatment effects over each subgroup's own covariate distribution. The standardized effect, $\tau^S(z)$, instead uses the population's marginal covariate distribution for every subgroup:
\begin{align}
\tau^J(z)
&=
\int\tau(z,x)\,dF_{X\mid Z=z}^\star(x),
\label{eq:subgroup-average-effect}\\
\tau^S(z)
&=
\int\tau(z,x)\,dF_X^\star(x).
\label{eq:standardized-effect-definition}
\end{align}
In the education example, $\tau^J(z)$ is the average policy effect among adults in education group $z$, whereas $\tau^S(z)$ is the effect standardized to the population's age distribution. For this period-$1$ target population, $\tau^J(z)=\tau(z)$ for $F_Z^\star$-almost every $z$.

The population-average effect averages treatment effects over each subgroup's own covariate distribution and then across subgroups according to their population shares,
\begin{equation}
\bar\tau
=
\int\tau(z,x)\,dF_{X,Z}^\star(x,z)
=
\int\tau^J(z)\,dF_Z^\star(z).
\label{eq:joint-average-effect}
\end{equation}
Standardized effects instead use the same covariate distribution for every subgroup. Their population-share-weighted average therefore need not equal $\bar\tau$, although the same aggregate estimate of $\bar\tau$ can be used to identify each standardized effect.

The corresponding observed-change averages use the same distributions,
\[
\begin{aligned}
\Delta^J(z)
&=\int\Delta(z,x)\,dF_{X\mid Z=z}^\star(x),\\
\Delta^S(z)
&=\int\Delta(z,x)\,dF_X^\star(x),\\
\bar\Delta
&=\int\Delta(z,x)\,dF_{X,Z}^\star(x,z).
\end{aligned}
\]
These averages hold the covariate distributions fixed across survey periods. Consequently, $\bar\Delta$ can differ from the average of the marginal subgroup changes in Section~\ref{subsec:relative-identification}, which also reflects changes in covariate composition between waves.

To recover changes at the target population's subgroup-covariate combinations, the joint distribution must satisfy the following overlap condition.

\begin{assumption}[Joint covariate overlap]
\label{ass:joint-overlap}
For every survey period $t$ used to estimate conditional changes,
\[
F_{X,Z}^\star\ll P_t^{X,Z}.
\]
\end{assumption}

Assumption~\ref{ass:joint-overlap} requires every set of subgroup-covariate combinations with positive target probability to have positive probability in each survey population used, while allowing their frequencies to change. It therefore makes the conditional changes available at the combinations used to calculate subgroup-average effects. To identify effects from those changes, the next assumption allows untreated changes to vary with $X$ but requires them to be equal across subgroups at a given $X$.

\begin{assumption}[Conditional relative parallel trends]
\label{ass:conditional-relative-parallel-trends}
For $F_{X,Z}^\star$-almost every $(x,z)$,
\begin{equation}
\mu_1^0(z,x)-\mu_0^0(z,x)=m(x)
\label{eq:conditional-relative-parallel-trends}
\end{equation}
for an unrestricted measurable function $m$.
\end{assumption}

Assumption~\ref{ass:conditional-relative-parallel-trends} permits, for example, coverage to change differently by age as long as the untreated change is the same across education groups at each age. It can therefore hold even when marginal relative parallel trends fails because the groups have different age compositions. An unrelated change in insurance coverage that affected education groups differently at the same age would violate this assumption.

\medskip
\noindent\emph{Standardized effects.} To identify effects at a common covariate distribution, the data must also cover the combinations to which that standardization assigns weight. These combinations follow the product distribution $F_X^\star\otimes F_Z^\star$, which gives every subgroup the same covariate distribution. This requires a further overlap assumption.

\begin{assumption}[Standardization overlap]
\label{ass:standardization-overlap}
The product target distribution satisfies
\[
F_X^\star\otimes F_Z^\star\ll P_t^{X,Z},
\qquad t=0,1.
\]
\end{assumption}

For categorical $Z$, Assumption~\ref{ass:standardization-overlap} requires any covariate set with positive probability under $F_X^\star$ to have positive probability within every target subgroup in both periods. In the education example, the common age distribution must therefore be represented within each education group. Together with joint overlap, which identifies the average change under the population's actual joint distribution, this allows both observed-change averages to be used. Under these overlap conditions and Assumption~\ref{ass:conditional-relative-parallel-trends},
\[
\Delta^S(z)=\int m(x)\,dF_X^\star(x)+\tau^S(z),
\qquad
\bar\Delta=\int m(x)\,dF_X^\star(x)+\bar\tau.
\]
Because the two expressions contain the same marginal average of $m(X)$, subtraction removes that unknown untreated change. Combining the resulting contrast with the population-average effect in \cref{eq:joint-average-effect} gives the following identification result.

\begin{proposition}[Identification of standardized heterogeneous effects]
\label{prop:standardized-effects}
Under \cref{ass:estimand-compatibility,ass:joint-overlap,ass:conditional-relative-parallel-trends,ass:standardization-overlap}, the standardized effect is identified for $F_Z^\star$-almost every $z$ as
\begin{equation}
\tau^S(z)=\bar\tau+\Delta^S(z)-\bar\Delta.
\label{eq:standardized-anchored-effect}
\end{equation}
\end{proposition}

\noindent\emph{Proof.} In Appendix~\ref{app:proofs}.\par\smallskip

Proposition~\ref{prop:standardized-effects} shows that subgroup effects standardized to the target population's marginal covariate distribution can be identified using the population-average effect for that same population. The aggregate comparative-case analysis need not itself be standardized, which is particularly useful when comparison jurisdictions report only overall outcomes. This is possible because the standardized subgroup change and the population-average change contain the same average untreated change, which cancels when they are differenced.

\medskip
\noindent\emph{Subgroup-average effects.} When treatment effects do not vary further with $X$, the standardized and subgroup-average effects coincide, $\tau^J(z)=\tau^S(z)$, so covariate adjustment does not change the subgroup effect being estimated. When treatment effects also vary with $X$, however, the two effects may differ. Identifying the subgroup-average effect then requires the subgroup's average untreated change. Under conditional relative parallel trends (Assumption~\ref{ass:conditional-relative-parallel-trends}), the observed change satisfies
\[
\Delta^J(z)-\tau^J(z)
=
\int m(x)\,dF_{X\mid Z=z}^\star(x).
\]
The aggregate effect identifies $\int m(x)\,dF_X^\star(x)=\bar\Delta-\bar\tau$, which averages untreated changes over the whole target population. The question is whether this also determines the average untreated change within each subgroup. The next proposition establishes that it need not.

\begin{proposition}[Nonidentification of subgroup-average effects]
\label{prop:subgroup-nonidentification}
Suppose that \cref{ass:joint-overlap,ass:conditional-relative-parallel-trends,ass:standardization-overlap} hold and that $F_{X\mid Z=z}^\star\neq F_X^\star$ on a set of $z$ values with positive $F_Z^\star$ probability. With $m(x)$ and $\tau(z,x)$ otherwise unrestricted, the conditional changes $\Delta(z,x)$, the target distribution $F_{X,Z}^\star$, and a compatible population-average effect $\bar\tau$ do not generally identify the subgroup-average effect function $z\mapsto\tau^J(z)$, even up to $F_Z^\star$-null sets.
\end{proposition}

\noindent\emph{Proof.} In Appendix~\ref{app:proofs}.\par\smallskip

Proposition~\ref{prop:subgroup-nonidentification} shows that identifying the standardized effects does not generally identify the average effect among members of each subgroup. Different patterns of untreated changes across $X$ can produce the same observed changes and population-average effect, yet imply different subgroup-average effects. This ambiguity can remain even when those observed changes and the aggregate effect are known exactly, so recovering the subgroup-average effects requires further information about the covariate-specific untreated changes.

\medskip
\noindent\emph{Covariate-specific aggregate effects.} The missing information can be obtained from aggregate comparative-case analyses conducted separately within covariate categories. Consider categorical $X$ with finite target-population support $\mathcal X$. For each $x\in\mathcal X$, averaging over the distribution of $Z$ within that category gives the category-average effect and observed change:
\[
\begin{aligned}
\bar\tau_x
&=\int\tau(z,x)\,dF_{Z\mid X=x}^\star(z)
=\E_1[Y_1(1)-Y_1(0)\mid X=x],\\
\bar\Delta_x
&=\int\Delta(z,x)\,dF_{Z\mid X=x}^\star(z).
\end{aligned}
\]
Under joint overlap, $\bar\Delta_x$ is identified from the treated-jurisdiction cross-sections. The next assumption specifies the additional causal information required from the aggregate analyses.

\begin{assumption}[Covariate-specific aggregate identification]
\label{ass:covariate-aggregate-compatibility}
The categories of $X$ with positive target probability form a finite set $\mathcal X$. For each $x\in\mathcal X$, an aggregate comparative-case design identifies the category-specific population-average effect $\bar\tau_x$ defined above.
\end{assumption}

Assumption~\ref{ass:covariate-aggregate-compatibility} requires each category-specific analysis to identify the effect for the target population and post-treatment period represented in the within-jurisdiction analysis. Under conditional relative parallel trends, combining that effect with the observed category change gives $m(x)=\bar\Delta_x-\bar\tau_x$. The untreated change can then be removed from each subgroup's observed change within the category. Writing $\Pr^\star$ for probability under $F_{X,Z}^\star$ and $p_x(z)=\Pr^\star(X=x\mid Z=z)$ for the target category share within subgroup $z$, the following result identifies both the conditional effects and their subgroup averages.

\begin{proposition}[Identification with covariate-specific aggregate effects]
\label{prop:covariate-aggregate-effects}
Under \cref{ass:joint-overlap,ass:conditional-relative-parallel-trends,ass:covariate-aggregate-compatibility}, $\tau(z,x)$ is identified for $F_{X,Z}^\star$-almost every $(x,z)$ as
\begin{equation}
\tau(z,x)=\Delta(z,x)-\bar\Delta_x+\bar\tau_x.
\label{eq:covariate-specific-effect}
\end{equation}
Consequently, the subgroup-average effects are identified for $F_Z^\star$-almost every $z$ as
\begin{equation}
\tau^J(z)
=
\sum_{x:p_x(z)>0}
 p_x(z)\left\{\Delta(z,x)-\bar\Delta_x+\bar\tau_x\right\}.
\label{eq:covariate-anchored-effect}
\end{equation}
\end{proposition}

\noindent\emph{Proof.} In Appendix~\ref{app:proofs}.\par\smallskip

Proposition~\ref{prop:covariate-aggregate-effects} shows how category-specific aggregate effects resolve the ambiguity in Proposition~\ref{prop:subgroup-nonidentification}. For example, aggregate effects for joint age-sex categories can be combined with treated-jurisdiction education data to recover each education group's average effect under its own age-sex distribution. This retains conditional relative parallel trends, imposed within each joint age-sex category, while requiring education data only in the treated jurisdiction. Because each group is averaged only over covariate categories that it contains, the result requires joint overlap without requiring standardization overlap.

\subsection{Using pretreatment trends}
\label{subsec:pretreatment-trends}

Further identification results become available with repeated cross-sections from multiple pretreatment periods. These observations can be used to estimate earlier differences in subgroup trends and, under assumptions about their continuation, adjust the treatment-period comparison. I then show how pretreatment covariate trends can identify subgroup-average effects when only a single population-average aggregate effect is available.

Let periods $-1$ and $0$ precede treatment and period $1$ follow treatment, with equal spacing. The observed outcomes satisfy $Y_s=Y_s(0)$ in pretreatment periods and $Y_1=Y_1(1)$. All averages retain the fixed period-$1$ target distributions. Assumption~\ref{ass:subgroup-overlap} applies to every period used in the marginal analysis, and Assumption~\ref{ass:joint-overlap} applies to every period used in the covariate-adjusted analysis. For the two results below, these periods are $-1,0,1$.

\medskip
\noindent\emph{Differential subgroup trends.} Consider the marginal effects and observed changes defined in Section~\ref{subsec:relative-identification}. To allow untreated changes to differ across subgroups, define each subgroup's deviation from the population-average untreated change as
\begin{equation}
r_s(z)
=
\mu_s^0(z)-\mu_{s-1}^0(z)
-
\int\{\mu_s^0(u)-\mu_{s-1}^0(u)\}\,dF_Z^\star(u).
\label{eq:differential-subgroup-trend}
\end{equation}
Under subgroup overlap, the two pretreatment cross-sections identify $r_0(z)$, but the corresponding treatment-transition change $r_1(z)$ remains counterfactual. It enters the observed comparison through the identity
\begin{equation}
\tau(z)-\bar\tau
=
\Delta(z)-\bar\Delta-r_1(z).
\label{eq:trend-adjusted-heterogeneity}
\end{equation}
Assumption~\ref{ass:relative-parallel-trends} identifies the contrast by setting $r_1(z)=0$. When a pretreatment transition is observed, the latest differential change instead provides a value that can be carried forward under a stability assumption, as in sequential difference-in-differences arguments \citep{egami2023multiple}.

\begin{assumption}[Stability of differential subgroup trends]
\label{ass:stable-subgroup-trends}
For the equally spaced transitions from period $-1$ to $0$ and from period $0$ to $1$,
\[
r_1(z)=r_0(z)
\]
for $F_Z^\star$-almost every $z$.
\end{assumption}

Substituting the observed pretreatment differential change for $r_1(z)$ in \cref{eq:trend-adjusted-heterogeneity} identifies the adjusted contrast. Combining it with the population-average effect gives the subgroup effects in the following proposition.

\begin{proposition}[Identification with differential subgroup trends]
\label{prop:pretrend-subgroup-effects}
Under \cref{ass:subgroup-overlap,ass:estimand-compatibility,ass:stable-subgroup-trends}, the subgroup effects are identified for $F_Z^\star$-almost every $z$ as
\begin{equation}
\tau(z)=\bar\tau+\Delta(z)-\bar\Delta-r_0(z).
\label{eq:pretrend-anchored-effect}
\end{equation}
\end{proposition}

\noindent\emph{Proof.} In Appendix~\ref{app:proofs}.\par\smallskip

Proposition~\ref{prop:pretrend-subgroup-effects} shows that subgroup effects can still be identified when untreated changes differ across subgroups, provided their deviations from the population-average change persist through the treatment transition. For example, if an education group had been gaining coverage relative to the population before a reform, that gain is carried into its untreated reform-period change and removed from the observed comparison. The assumption allows the population-average untreated change to shift between transitions, but would be violated by a concurrent shock that changed an education group's untreated trend relative to the population.

\medskip
\noindent\emph{Covariate-dependent untreated trends.} Proposition~\ref{prop:subgroup-nonidentification} shows that the population-average effect does not determine how untreated changes vary with $X$, leaving subgroup-average effects unidentified. Pretreatment observations can provide information about these differences. As in Section~\ref{subsec:covariate-adjustment}, the target remains the treated jurisdiction's period-$1$ joint distribution $F_{X,Z}^\star=P_1^{X,Z}$. Its marginal $F_X^\star$ is used to average untreated changes across $X$, while $F_{X\mid Z=z}^\star$ is used to average them within subgroup $z$. To use the pretreatment information, I extend conditional relative parallel trends (Assumption~\ref{ass:conditional-relative-parallel-trends}) to the transition from period $-1$ to $0$ and require that differences in untreated changes across values of $X$ remain stable between this transition and the transition from $0$ to $1$.

\begin{assumption}[Stability of relative covariate trends]
\label{ass:stable-covariate-trends}
For the equally spaced transitions ending in $s=0,1$, there are measurable functions $m_s$ such that
\[
\mu_s^0(z,x)-\mu_{s-1}^0(z,x)=m_s(x),
\qquad s\in\{0,1\},
\]
for $F_{X,Z}^\star$-almost every $(x,z)$. In addition,
\[
m_1(x)=m_0(x)+c
\]
for $F_X^\star$-almost every $x$, where $c$ is an unrestricted scalar.
\end{assumption}

Assumption~\ref{ass:stable-covariate-trends} permits the overall pace of untreated change to differ between transitions, while preserving its pattern across covariate values. For example, age-specific coverage changes may all shift by the same amount, provided that their differences across ages remain stable. A change in that relative age pattern would violate the assumption. For each transition, define the deviation at $X=x$ from the population-average untreated change as
\begin{equation}
q_s(x)=m_s(x)-\int m_s(u)\,dF_X^\star(u),
\qquad s\in\{0,1\}.
\label{eq:relative-covariate-trend}
\end{equation}
Centering removes the common shift $c$, so the stability condition is equivalent to $q_1(x)=q_0(x)$ for $F_X^\star$-almost every $x$. Under joint overlap, the pretreatment observations identify $m_0(x)$ and hence $q_0(x)$. Substituting this observed relative pattern for its counterfactual treatment-period counterpart gives the following result.

\begin{samepage}
\begin{proposition}[Identification of subgroup-average effects from pretreatment observations]
\label{prop:pretrend-subgroup-average-effects}
Under \cref{ass:estimand-compatibility,ass:joint-overlap,ass:stable-covariate-trends}, the subgroup-average effects are identified for $F_Z^\star$-almost every $z$ as
\begin{equation}
\tau^J(z)
=
\bar\tau+\Delta^J(z)-\bar\Delta
-
\int q_0(x)\,dF_{X\mid Z=z}^\star(x).
\label{eq:pretrend-subgroup-average-effect}
\end{equation}
\end{proposition}
\end{samepage}

\noindent\emph{Proof.} In Appendix~\ref{app:proofs}.\par\smallskip

Proposition~\ref{prop:pretrend-subgroup-average-effects} shows how to identify subgroup-average effects by combining a single population-average effect with pretreatment observations under the stability of relative covariate trends (Assumption~\ref{ass:stable-covariate-trends}). Subtracting $\bar\tau$ from the target-average observed change $\bar\Delta$ gives the population-average untreated change. The pretreatment observations identify how untreated changes differ across $X$, and the stability assumption allows these differences to be carried forward to the treatment transition. Averaging the resulting deviations over each subgroup's own covariate distribution then provides the adjustment needed to identify its average policy effect. Unlike the approach in Proposition~\ref{prop:covariate-aggregate-effects}, this result does not require category-specific aggregate effects.

The two pretreatment results use stability assumptions to determine the missing differential change, but the identification formulas also accommodate other specified extrapolations. With more than two pretreatment periods, for example, a linear model for successive differential changes gives $r_1(z)=2r_0(z)-r_{-1}(z)$ and requires at least three pretreatment cross-sections. The same form can be used for $q_s(x)$, provided that the specified model continues over the treatment transition and overlap holds in every period used. The extrapolated $r_1(z)$ or $q_1(x)$ then replaces $r_0(z)$ or $q_0(x)$ in \cref{eq:pretrend-anchored-effect,eq:pretrend-subgroup-average-effect}. When assumptions bound the differential change instead of determining it exactly, the corresponding treatment effects are bounded as well \citep{rambachan2023credible}.
\section{Estimation and inference}
\label{sec:estimation}

\subsection{Estimation}
\label{subsec:estimation}

For the subgroup effects identified in Proposition~\ref{prop:anchored-identification}, an estimate of the population-average effect can be combined with estimated within-jurisdiction changes to obtain
\[
\widehat\tau(z)
=
\widehat{\bar\tau}
+
\widehat\Delta(z)
-
\int\widehat\Delta(u)\,d\widehat F_Z^\star(u).
\]
For categorical $Z$, $\widehat\Delta(z)$ can be estimated from group-specific pre--post changes. For continuous $Z$, it can instead be estimated from a model for the conditional mean change.

Under the conditions of Proposition~\ref{prop:standardized-effects}, when treatment effects depend on $Z$ but not further on $X$, the additive representation $\Delta(z,x)=m(x)+\tau(z)$ can be fitted to the repeated cross-sections. For categorical $Z$, a repeated-cross-sectional regression with group-by-post interactions and post-period changes allowed to vary with $X$ provides a convenient implementation. If $\widehat a(z)$ denotes the fitted $Z$-specific component, the corresponding effect estimate is $\widehat{\bar\tau}+\widehat a(z)-\int\widehat a(u)\,d\widehat F_Z^\star(u)$.

For standardized effects, the conditional changes are averaged over the common marginal covariate distribution and the joint target distribution, respectively. If $\widehat\Delta(z,x)$ denotes an estimate of the conditional change, the standardized estimator is
\begin{equation}
\widehat\tau^S(z)
=
\widehat{\bar\tau}
+
\int\widehat\Delta(z,x)\,d\widehat F_X^\star(x)
-
\int\widehat\Delta(u,x)\,d\widehat F_{X,Z}^\star(x,u).
\label{eq:estimated-standardized-effect}
\end{equation}
For Proposition~\ref{prop:covariate-aggregate-effects}, let $\widehat{\bar\tau}_x$ denote the category-specific aggregate estimate under Assumption~\ref{ass:covariate-aggregate-compatibility}, and let $\widehat p_x(z)$ estimate the corresponding target-population category share. The plug-in estimator is
\begin{equation}
\widehat\tau^J(z)
=
\widehat\Delta^J(z)
+
\sum_{x\in\mathcal X}
\widehat p_x(z)\left\{\widehat{\bar\tau}_x-\widehat{\bar\Delta}_x\right\},
\label{eq:estimated-covariate-anchored-effect}
\end{equation}
where $\widehat\Delta^J(z)$ and $\widehat{\bar\Delta}_x$ average the fitted conditional changes over the estimated target distributions within subgroup $z$ and category $x$, respectively.

Plug-in estimators based on Propositions~\ref{prop:pretrend-subgroup-effects} and~\ref{prop:pretrend-subgroup-average-effects} replace $r_0(z)$ and $q_0(x)$ with their pretreatment estimates, alongside estimates of the observed changes and target distributions in \cref{eq:pretrend-anchored-effect,eq:pretrend-subgroup-average-effect}. Under an alternative parametric extrapolation, the fitted treatment-transition differential trend replaces the most recent pretreatment value.

\subsection{Inference}
\label{subsec:inference}

For estimators based on a single population-average effect, write $\widehat\tau(z)=\widehat{\bar\tau}+\widehat c(z)$, where $\widehat c(z)$ is the within-jurisdiction contribution. When the two components are estimated from independent samples, a natural variance estimator is
\[
\widehat{\Var}\{\widehat\tau(z)\}
=
\widehat{\Var}(\widehat{\bar\tau})
+
\widehat{\Var}\{\widehat c(z)\}.
\]
Pointwise confidence intervals can then be obtained using a normal approximation. For standardized effects, $\widehat c(z)$ is the difference between two averages obtained from the same fitted conditional-change function, so its variance should include their covariance as well as uncertainty in the estimated target distribution. A bootstrap can account for these sources of uncertainty by refitting the conditional-change model, re-estimating the target distribution, and recomputing both averages within each resample. The same principle applies when the within-jurisdiction component includes an estimated pretreatment-trend adjustment.

When bootstrap distributions are available for both components, uncertainty can instead be propagated through component bootstrap errors. Let $e_A^{*(b)}$ denote a bootstrap error for $\widehat{\bar\tau}$ and $e_C^{*(b)}(z)$ the corresponding error for $\widehat c(z)$. When the two components are estimated from independent samples, independent bootstrap errors can be combined as
\begin{equation}
\widehat\tau^{*(b)}(z)
=
\widehat\tau(z)
+
e_A^{*(b_A)}
+
e_C^{*(b_C)}(z).
\label{eq:independent-bootstrap}
\end{equation}
If the same data contribute to both components, their covariance must also be accounted for. When both estimates can be recomputed from a common resample, a joint bootstrap provides one way to preserve this dependence.

When the aggregate comparative-case estimator supplies a confidence interval but no bootstrap distribution, component intervals can still be combined conservatively. If $[L_A,U_A]$ is an interval for $\bar\tau$ with coverage at least $1-\alpha_A$ and $[L_C(z),U_C(z)]$ is an interval for $c(z)$ with coverage at least $1-\alpha_C$, their endpoint sum $[L_A+L_C(z),\,U_A+U_C(z)]$ has pointwise coverage at least $1-\alpha_A-\alpha_C$ by Bonferroni's inequality, without an independence assumption.

For \cref{eq:estimated-covariate-anchored-effect}, inference also accounts for covariance among the category-specific aggregate estimates. When the aggregate and within-jurisdiction components are estimated independently, draws from a joint bootstrap distribution of the aggregate estimates can be paired with independent within-jurisdiction resamples. In each within-jurisdiction resample, the subgroup shares and observed changes are re-estimated before recomputing the combined estimate using \cref{eq:estimated-covariate-anchored-effect}. If the components share observations or an estimated target distribution, resampling must preserve the induced dependence.

\section{Numerical illustrations}
\label{sec:numerical}

I use two Monte Carlo experiments to examine estimation and inference in the two-period setting. The first evaluates the bias, interval coverage, and precision of the combined subgroup-effect estimator as the aggregate pretreatment history and repeated-cross-sectional sample sizes vary. Aggregate outcomes for 40 jurisdictions follow a rank-two interactive fixed-effects model, with one jurisdiction treated in the final period. Independent repeated cross-sections within the treated jurisdiction contain three population groups with shares $(0.35,0.30,0.35)$ and treatment effects $(0.35,0.25,0.15)$, giving a population-average effect of 0.25. I estimate the population-average effect using the generalized synthetic control method \citep{xu2017generalized}, with the factor rank fixed at its true value. The designs combine $T_0\in\{10,30\}$ pretreatment periods with $n\in\{500,2500\}$ observations per survey wave; complete data-generating processes and subgroup-specific results are given in \OASection{sec:simulations}.

\begin{table}[t]
\centering
\caption{Finite-sample performance of the combined estimator}
\label{tab:simulation}
\begin{threeparttable}
\small
\begin{tabular}{S[table-format=2.0] S[table-format=4.0] S[table-format=1.3] S[table-format=1.3] c}
\toprule
{$T_0$} & {$n$ per wave} & {Weighted mean $|$bias$|$} & {Weighted mean RMSE} & {95\% coverage range} \\
\midrule
10 & 500 & 0.005 & 0.275 & 0.940--0.964 \\
10 & 2500 & 0.001 & 0.271 & 0.956--0.964 \\
30 & 500 & 0.004 & 0.237 & 0.954--0.968 \\
30 & 2500 & 0.009 & 0.223 & 0.946--0.958 \\
\bottomrule
\end{tabular}
\begin{tablenotes}[flushleft]
\footnotesize
\item \textit{Notes:} Bias and RMSE refer to the three subgroup-specific combined treatment-effect estimators. Weighted mean $|$bias$|$ is the target-population-share-weighted mean of the absolute subgroup biases, and weighted mean root mean squared error (RMSE) is the corresponding weighted mean of the subgroup RMSEs. Bias and RMSE are based on 1,000 Monte Carlo replications. Coverage reports the range across the three subgroup-specific nominal 95\% Wald intervals and is based on 500 replications.
\end{tablenotes}
\end{threeparttable}
\end{table}

The combined estimator shows little finite-sample bias, and coverage remains close to the nominal level across all four designs (Table~\ref{tab:simulation}). The aggregate component accounts for the larger share of variance in this design, so increasing the pretreatment history from 10 to 30 periods reduces root mean squared error more than increasing the repeated-cross-sectional sample from 500 to 2,500 observations. The largest absolute subgroup bias is \SimMaxBias, and coverage of nominal 95\% intervals ranges from \SimCoverageMin\ to \SimCoverageMax.

The second experiment investigates whether standardization recovers effects defined at a common $X$ distribution when $X$ affects both untreated trends and treatment effects. The distribution of a continuous covariate $X$ differs across the three groups, untreated changes include the term $0.30X$, and treatment effects include $0.15X$. The standardized effects are 0.35, 0.25, and 0.15. Without standardization, the corresponding population values of the combined estimator are 0.08, 0.25, and 0.42. The unadjusted construction therefore reverses the ordering of the standardized effects in this design. Figure~\ref{fig:standardized-simulation} compares these quantities with Monte Carlo mean estimates from 1,000 replications, using $T_0=20$ and 2,000 observations per survey wave. The unadjusted estimates reproduce the reversed ordering, whereas the standardized estimates recover the target effects with a maximum absolute bias of \StdMaxBias.

\begin{figure}[t]
\centering
\includegraphics[width=0.72\textwidth]{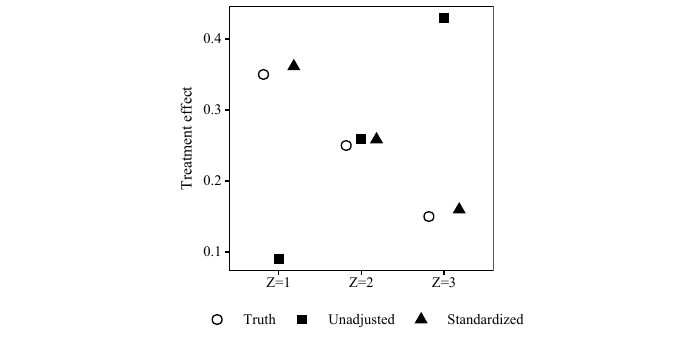}
\caption{Covariate standardization. Symbols show the target standardized effects and Monte Carlo mean estimates of the unadjusted and standardized estimators across 1,000 replications. The design uses $T_0=20$ and $n=2{,}000$ observations per survey wave.}
\label{fig:standardized-simulation}
\end{figure}

\section{Massachusetts health reform}
\label{sec:massachusetts}

Massachusetts's health reform sought near-universal insurance coverage through expanded subsidized coverage and an individual mandate \citep{long2008road}. I illustrate the framework using this reform. The Massachusetts Health Reform Survey (MHRS) provides detailed repeated cross-sections within the state around the main implementation of the reform, including education, age, sex, and current insurance coverage. I construct a state-level comparison panel from insurance coverage rates in the Behavioral Risk Factor Surveillance System (BRFSS). The combined estimator requires only aggregate coverage rates from the comparison states. Because the underlying BRFSS microdata are also available, the application also permits direct education-specific comparative-case analyses that use individual-level information from those states. Additional data and estimation details are given in \OASection{sec:mass-data}.

Massachusetts enacted its health-insurance reform in 2006, but the major components took effect during 2007 \citep{long2008road}. Consistent with previous evaluations, I use the 2006 MHRS as the baseline wave: although some provisions were already in place, the survey was fielded before the main components of the reform took effect. For the aggregate analysis, I date the policy transition from 2007Q1 and evaluate the effect in 2007Q4, matching the timing of the 2007 MHRS. For both components, the target population is Massachusetts adults aged 18--64 and the outcome is current health insurance coverage. I use the quarterly BRFSS panel to estimate the population-average effect and the 2006 and 2007 MHRS to estimate the within-jurisdiction treatment-effect contrasts. The 2007 MHRS defines the target covariate distribution.

For the aggregate comparative-case analysis, I use the generalized synthetic control method \citep{xu2017generalized}, which models untreated outcomes with unit and period effects and a low-rank interactive fixed-effects component. I select the number of latent factors by rolling pretreatment cross-validation, as implemented in the \texttt{fect} package for R \citep{liu2024practical}. The selected specification contains no interactive factors, so the fitted counterfactual reduces to an additive state-and-period outcome model, closely related to the imputation-based difference-in-differences specification in \citet{borusyak2024revisiting}. For inference on the population-average effect, I use the parametric bootstrap procedure in \citet{xu2017generalized}.

The fitted untreated trajectory closely tracks the observed pretreatment series (Figure~\ref{fig:massachusetts-application}, panel~(a)). In 2007Q4, observed coverage exceeds the estimated counterfactual by \AggregateATT\ percentage points (95\% confidence interval [\AggregateLower, \AggregateUpper]).

I examine heterogeneity across three education groups: high school education or less, some college or an associate degree, and a bachelor's degree or higher. Age and sex differ across education groups and may affect both untreated changes in coverage and responses to the reform. I therefore allow the conditional pre--post change to vary jointly with education, age, and sex and standardize each education group to the same age-sex distribution from the 2007 MHRS. The within-Massachusetts identification assumption is that, absent the 2007 policy transition, coverage would have changed equally across education groups of the same age and sex. For inference on the within-jurisdiction component, I use a survey bootstrap that respects the MHRS sampling design (see \OASection{sec:mass-data} for details). I combine independent bootstrap draws from the two components because the BRFSS panel and MHRS are based on separate survey samples.

Panel~(b) of Figure~\ref{fig:massachusetts-application} reports each education group's standardized effect relative to the population-average effect. In ascending order of education, these treatment-effect contrasts are \HSLessRelative, \SomeCollegeRelative, and \BAPlusRelative\ percentage points. The open circles show the corresponding unadjusted relative changes. Because each education group is standardized to the same age-sex distribution, these contrasts need not average to zero across the observed education shares.

Panel~(c) of Figure~\ref{fig:massachusetts-application} combines the population-average effect with the within-Massachusetts treatment-effect contrasts. The resulting standardized effects on insurance coverage are \HSLessCombined, \SomeCollegeCombined, and \BAPlusCombined\ percentage points in ascending order of education. For the lowest education group, the fitted standardized change is \HSLessStandardizedChange\ percentage points, while the fitted change averaged over the target joint distribution of education, age, and sex is \TargetAverageChange\ points. Their difference, \HSLessRelative\ points, is the standardized treatment-effect contrast. Combining this contrast with the population-average estimate gives a standardized effect of \HSLessCombined\ points. The estimated difference between the lowest and highest education groups is \LowHighContrast\ percentage points (95\% confidence interval [\LowHighLower, \LowHighUpper]), and an omnibus test of equal standardized effects gives $\chi^2(2)=\EducationOmnibus$ ($p<0.001$).

Panel~(c) also reports direct education-specific comparative-case estimates from BRFSS state panels standardized to the same age-sex target distribution. These analyses use individual-level BRFSS data from Massachusetts and the comparison states, information that is not required by the combined estimator. The resulting point estimates are similar: \BRFSSLow, \BRFSSMiddle, and \BRFSSHigh\ percentage points from the education-specific BRFSS panels, compared with \HSLessCombined, \SomeCollegeCombined, and \BAPlusCombined\ points from the combined estimator. Both approaches therefore produce the same broad pattern, with the largest estimated coverage increase in the lowest education group and smaller effects in the other two groups.

\begin{figure}[t]
\centering
\includegraphics[width=0.98\textwidth]{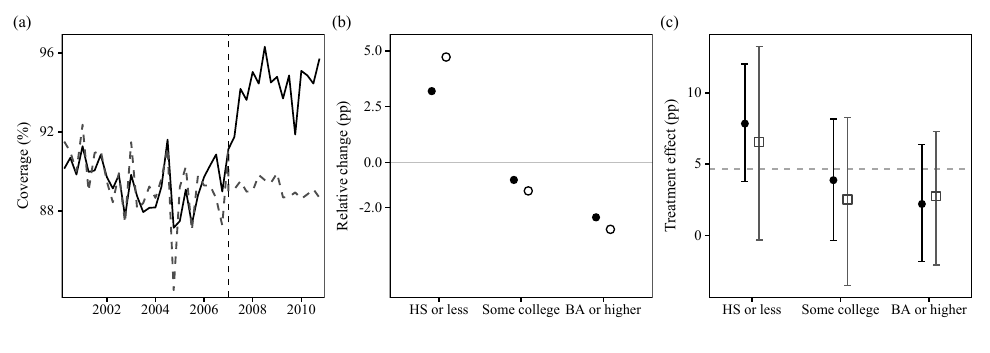}
\caption{Massachusetts health reform. Panel (a) shows observed quarterly insurance coverage among Massachusetts adults aged 18--64 (solid line) and the estimated untreated trajectory (dashed line); the vertical dashed line marks the 2007Q1 treatment onset. Panel (b) shows education-specific changes relative to the target-population change before (open circles) and after (filled circles) standardization to the 2007 Massachusetts Health Reform Survey age-sex distribution. Panel (c) shows the combined estimates (filled circles) together with estimates from separate education-specific BRFSS state panels (open squares). Vertical bars are pointwise 95\% confidence intervals, and the dashed horizontal line denotes the estimated population-average effect.}
\label{fig:massachusetts-application}
\end{figure}

\section{Discussion}
\label{sec:discussion}

This article develops an identification framework for heterogeneous policy effects when individual-level data are available only in the treated jurisdiction. Under relative parallel trends, subgroup differences in observed changes and a population-average effect from aggregate comparisons identify subgroup effects. The within-jurisdiction analysis requires only one pretreatment cross-section in this case.

The covariate results show why an effect at a common covariate distribution need not answer the question of how the policy affected a subgroup's members on average. These questions can require different information even when the population-average effect has been identified. I show how covariate-specific aggregate effects or pretreatment observations under trend-stability assumptions can supply the additional information needed for subgroup-average effects. Researchers can use these results to determine whether the available data identify the effect relevant to their substantive question and what further information or assumptions would be needed.

\clearpage
\appendix

\titleformat{\section}[block]
  {\centering\large\bfseries}
  {Appendix \thesection.}{0.65em}{}

\section{Proofs}
\label{app:proofs}

\noindent\emph{Technical conventions.} Throughout the paper, covariates take values in standard Borel spaces, including finite-dimensional real spaces and finite or countable sets. Conditional distributions are regular versions, and conditional means are measurable versions. All expectations and target-weighted integrals used in the paper are assumed to exist and be finite.

\medskip
\begin{proof}[Proof of Proposition~\ref{prop:anchored-identification}]
Under Assumption~\ref{ass:subgroup-overlap}, $F_Z^\star=P_1^Z\ll P_0^Z$, so the observed conditional means identify $\Delta(z)$ for $F_Z^\star$-almost every $z$ and its target average $\bar\Delta$. Under Assumption~\ref{ass:relative-parallel-trends}, $\Delta(z)=\delta+\tau(z)$ for $F_Z^\star$-almost every $z$. Averaging over $F_Z^\star$ and subtracting gives $\Delta(z)-\bar\Delta=\tau(z)-\bar\tau$ on the same almost-everywhere domain. Assumption~\ref{ass:estimand-compatibility} identifies $\bar\tau$, yielding \cref{eq:anchored-effect}.
\end{proof}

\medskip
\begin{proof}[Proof of Proposition~\ref{prop:standardized-effects}]
Assumption~\ref{ass:joint-overlap} identifies the conditional changes up to $F_{X,Z}^\star$-null sets and hence identifies $\bar\Delta$. Assumption~\ref{ass:standardization-overlap} also identifies their averages over $F_X^\star$ for $F_Z^\star$-almost every $z$. Under Assumption~\ref{ass:conditional-relative-parallel-trends},
\[
\Delta(z,x)=m(x)+\tau(z,x)
\]
for $F_{X,Z}^\star$-almost every $(x,z)$. Since $F_X^\star\otimes F_Z^\star\ll P_1^{X,Z}=F_{X,Z}^\star$, this equality holds under the product target measure as well. Integrability and iterated integration therefore give
\[
\Delta^S(z)
=
\int m(x)\,dF_X^\star(x)+\tau^S(z)
\]
for $F_Z^\star$-almost every $z$, while averaging over the joint target distribution gives
\[
\bar\Delta
=
\int m(x)\,dF_X^\star(x)+\bar\tau.
\]
The latter equality uses the fact that $F_X^\star$ is the $X$-marginal of $F_{X,Z}^\star$. Assumption~\ref{ass:estimand-compatibility} identifies $\bar\tau$. Subtraction yields \cref{eq:standardized-anchored-effect} for $F_Z^\star$-almost every $z$.
\end{proof}

\medskip
\begin{proof}[Proof of Proposition~\ref{prop:subgroup-nonidentification}]
Because $X$ takes values in a standard Borel space, there is a countable algebra $\mathcal A$ of measurable sets that determines probability measures on that space. The set of $z$ values for which $F_{X\mid Z=z}^\star\neq F_X^\star$ is, up to a target-null set, the union over $A\in\mathcal A$ of the measurable sets
\[
B_A=\{z:F_{X\mid Z=z}^\star(A)\neq F_X^\star(A)\}.
\]
By the proposition's premise and countability, there is an $A\in\mathcal A$ with $F_Z^\star(B_A)>0$. Fix this $A$ and let
\[
b(x)=\ind\{x\in A\}-F_X^\star(A).
\]
Then
\[
\int b(x)\,dF_X^\star(x)=0,
\qquad
\int b(x)\,dF_{X\mid Z=z}^\star(x)
=
F_{X\mid Z=z}^\star(A)-F_X^\star(A).
\]
The second expression is nonzero on $B_A$.

Keep the observed conditional outcome distributions and the covariate distributions fixed. For a nonzero scalar $\varepsilon$, replace the unobserved conditional mean $\mu_1^0(u,x)$ by $\mu_1^0(u,x)+\varepsilon b(x)$. The resulting functions are
\[
m_\varepsilon(x)=m(x)+\varepsilon b(x),
\qquad
\tau_\varepsilon(u,x)=\tau(u,x)-\varepsilon b(x).
\]
These alternatives preserve the conditional relative-parallel-trends assumption and the observed conditional changes. The target population-average effect remains unchanged because
\[
\bar\tau_\varepsilon-\bar\tau
=
-\varepsilon\int b(x)\,dF_X^\star(x)
=
0.
\]
An aggregate analysis identifying only $\bar\tau$ cannot distinguish these alternatives. However,
\[
\tau_\varepsilon^J(z)-\tau^J(z)
=
-\varepsilon\{F_{X\mid Z=z}^\star(A)-F_X^\star(A)\},
\]
which is nonzero on a set with positive $F_Z^\star$ probability. Thus the subgroup-average effect functions differ beyond a target-null set.

For bounded outcomes, take an admissible model whose conditional counterfactual means are uniformly separated from the outcome bounds and choose $|\varepsilon|$ sufficiently small. The perturbed means can then be realized by counterfactual outcome distributions within those bounds. This establishes nonidentification in general, including when $Z$ is continuous. For a categorical subgroup with positive target probability and a covariate distribution different from $F_X^\star$, choosing $A$ to distinguish that subgroup's distribution gives the corresponding nonidentification of its subgroup-average effect.
\end{proof}

\medskip
\begin{proof}[Proof of Proposition~\ref{prop:covariate-aggregate-effects}]
Under Assumption~\ref{ass:joint-overlap}, $\Delta(z,x)$ is identified for $F_{X,Z}^\star$-almost every $(x,z)$. Assumption~\ref{ass:conditional-relative-parallel-trends} gives $\Delta(z,x)=m(x)+\tau(z,x)$ on that domain. Since each $x\in\mathcal X$ has positive target probability, averaging over $F_{Z\mid X=x}^\star$ gives
\[
\bar\Delta_x=m(x)+\bar\tau_x.
\]
The observed conditional changes identify $\bar\Delta_x$, and Assumption~\ref{ass:covariate-aggregate-compatibility} identifies $\bar\tau_x$. Hence
\[
\Delta(z,x)-\bar\Delta_x+\bar\tau_x=\tau(z,x)
\]
for $F_{X,Z}^\star$-almost every $(x,z)$. Averaging over the target covariate distribution within subgroup $z$ yields \cref{eq:covariate-anchored-effect} for $F_Z^\star$-almost every $z$.
\end{proof}

\medskip
\begin{proof}[Proof of Proposition~\ref{prop:pretrend-subgroup-effects}]
Under Assumption~\ref{ass:subgroup-overlap}, applied to periods $-1,0,1$, the pretreatment conditional changes, $r_0(z)$, and the observed treatment-transition change $\Delta(z)$ are identified for $F_Z^\star$-almost every $z$. Their averages over $F_Z^\star$ are therefore identified as well. For the treatment transition,
\[
\Delta(z)=\mu_1^0(z)-\mu_0^0(z)+\tau(z).
\]
Averaging over the fixed target distribution and subtracting gives
\[
\Delta(z)-\bar\Delta=r_1(z)+\tau(z)-\bar\tau
\]
for $F_Z^\star$-almost every $z$. Assumption~\ref{ass:stable-subgroup-trends} sets $r_1(z)=r_0(z)$ on that domain, and Assumption~\ref{ass:estimand-compatibility} identifies $\bar\tau$. Rearranging yields \cref{eq:pretrend-anchored-effect}.
\end{proof}

\medskip
\begin{proof}[Proof of Proposition~\ref{prop:pretrend-subgroup-average-effects}]
Assumption~\ref{ass:joint-overlap}, applied to periods $-1,0,1$, identifies the observed conditional changes under the fixed target joint distribution in every period used. Under Assumption~\ref{ass:stable-covariate-trends}, averaging the pretreatment conditional change over $F_{Z\mid X=x}^\star$ identifies $m_0(x)$, and hence $q_0(x)$, for $F_X^\star$-almost every $x$. Centering $m_1(x)=m_0(x)+c$ gives
\[
q_1(x)
=
m_0(x)+c-\int\{m_0(u)+c\}\,dF_X^\star(u)
=
q_0(x),
\]
for $F_X^\star$-almost every $x$. For the treatment transition,
\[
\bar\Delta-\bar\tau=\int m_1(x)\,dF_X^\star(x).
\]
Disintegrating the joint target distribution therefore gives
\[
\begin{aligned}
\Delta^J(z)-\tau^J(z)
&=
\int m_1(x)\,dF_{X\mid Z=z}^\star(x)\\
&=
\bar\Delta-\bar\tau+\int q_0(x)\,dF_{X\mid Z=z}^\star(x),
\end{aligned}
\]
for $F_Z^\star$-almost every $z$. The observed-change averages are identified by joint overlap, the pretreatment observations identify the averaged $q_0$ term, and Assumption~\ref{ass:estimand-compatibility} identifies $\bar\tau$. Rearranging gives \cref{eq:pretrend-subgroup-average-effect}.
\end{proof}

\clearpage

\clearpage

\bibliographystyle{apalike}

\bibliography{references}

\clearpage
\begin{center}
{\LARGE\bfseries Online Appendix}
\end{center}
\vspace{1em}
\setcounter{section}{0}
\setcounter{subsection}{0}
\setcounter{equation}{0}
\setcounter{table}{0}
\setcounter{figure}{0}
\renewcommand{\thesection}{OA.\arabic{section}}
\renewcommand{\thesubsection}{\thesection.\arabic{subsection}}
\renewcommand{\theequation}{OA.\arabic{equation}}
\renewcommand{\thetable}{OA.\arabic{table}}
\renewcommand{\thefigure}{OA.\arabic{figure}}
\renewcommand{\theHsection}{OA.\arabic{section}}
\renewcommand{\theHsubsection}{OA.\arabic{section}.\arabic{subsection}}
\renewcommand{\theHequation}{OA.\arabic{equation}}
\renewcommand{\theHtable}{OA.\arabic{table}}
\renewcommand{\theHfigure}{OA.\arabic{figure}}
\titleformat{\section}[block]
  {\centering\large\bfseries}
  {\thesection.}{0.65em}{}

\section{Simulation designs and supplementary results}
\label{sec:simulations}

This section documents the two Monte Carlo designs and reports subgroup-specific results.

\subsection{Finite-sample estimation and inference}

Aggregate untreated outcomes are generated for $J=40$ jurisdictions, one of which is treated in the final period. Untreated outcomes follow the rank-two interactive-fixed-effects model
\begin{equation}
A_{jt}(0)
=
\alpha_j+\delta_t+\lambda_j^{\mathsf T}f_t+u_{jt},
\label{eq:oa-panel-dgp}
\end{equation}
where $\lambda_j$ and $f_t$ are two-dimensional. The unit effects satisfy $\alpha_j\sim N(0,0.4^2)$ and $\lambda_j\sim N(0,I_2)$. The two elements of $f_t$ follow independent stationary AR(1) processes with autoregressive coefficient 0.7 and stationary standard deviation 0.7. The common time effect $\delta_t$ follows a stationary AR(1) process with autoregressive coefficient 0.5 and stationary standard deviation 0.2, while $u_{jt}\sim N(0,0.2^2)$. The innovations in these processes are mutually independent. I discard the first 50 simulated periods as burn-in. In the treated jurisdiction, the population-average post-treatment effect is 0.25.

Within the treated jurisdiction, independent repeated cross-sections are drawn before and after treatment from a population with group shares $(0.35,0.30,0.35)$. Individual untreated outcomes follow
\begin{equation}
Y_{it}(0)
=
A_{1t}(0)+\gamma_{Z_i}+\varepsilon_{it},
\label{eq:oa-micro-dgp}
\end{equation}
where $(\gamma_1,\gamma_2,\gamma_3)=(-0.4,0,0.4)$ and $\varepsilon_{it}\sim N(0,1)$. Because the population-weighted mean of the group shifts is zero, the population mean of the individual-level outcome equals the treated jurisdiction's aggregate outcome. Treatment effects are $(0.35,0.25,0.15)$ across the three groups, with population average 0.25.

I combine $T_0\in\{10,30\}$ pretreatment periods with $n\in\{500,2500\}$ observations per repeated cross-section. I estimate the population-average effect using the generalized synthetic control method \citep{xu2017generalized}, with the factor rank fixed at two. The target group shares are estimated from the pooled pre- and post-treatment observations in each replication. Bias and root mean squared error (RMSE) are based on 1,000 Monte Carlo replications, and coverage is evaluated in 500. Within each coverage replication, I estimate uncertainty in the population-average effect using 299 empirical-residual parametric bootstrap draws under the fitted model and uncertainty in the within-jurisdiction component using 299 respondent resamples drawn separately within each wave. I add the component variance estimates and construct nominal 95\% Wald intervals using the standard-normal critical value. Table~\ref{tab:oa-finite-sample-subgroups} reports the subgroup-specific results.

\begin{table}[t]
\centering
\caption{Subgroup-specific finite-sample performance}
\label{tab:oa-finite-sample-subgroups}
\begin{threeparttable}
\small
\begin{tabular}{S[table-format=2.0] S[table-format=4.0] l S[table-format=-1.3] S[table-format=1.3] S[table-format=1.3]}
\toprule
{$T_0$} & {$n$ per wave} & {Group} & {Bias} & {RMSE} & {95\% coverage} \\
\midrule
10 & 500 & $Z=1$ & -0.008 & 0.278 & 0.964 \\
10 & 500 & $Z=2$ & 0.000 & 0.276 & 0.940 \\
10 & 500 & $Z=3$ & -0.005 & 0.270 & 0.952 \\
10 & 2500 & $Z=1$ & 0.002 & 0.272 & 0.956 \\
10 & 2500 & $Z=2$ & -0.000 & 0.270 & 0.962 \\
10 & 2500 & $Z=3$ & 0.000 & 0.270 & 0.964 \\
30 & 500 & $Z=1$ & -0.005 & 0.233 & 0.958 \\
30 & 500 & $Z=2$ & 0.003 & 0.241 & 0.954 \\
30 & 500 & $Z=3$ & -0.003 & 0.238 & 0.968 \\
30 & 2500 & $Z=1$ & -0.008 & 0.224 & 0.946 \\
30 & 2500 & $Z=2$ & -0.011 & 0.222 & 0.956 \\
30 & 2500 & $Z=3$ & -0.010 & 0.224 & 0.958 \\
\bottomrule
\end{tabular}
\begin{tablenotes}[flushleft]
\footnotesize
\item \textit{Notes:} Bias and RMSE are based on 1,000 Monte Carlo replications. Coverage refers to nominal 95\% Wald intervals and is based on 500 replications.
\end{tablenotes}
\end{threeparttable}
\end{table}

\subsection{Covariate standardization}

The second experiment retains the aggregate process, group shares, baseline group shifts, and individual disturbance from the first design, but adds a continuous covariate whose distribution differs across groups:
\begin{equation}
X\mid Z=g
\sim
N(\mu_g,1),
\qquad
(\mu_1,\mu_2,\mu_3)=(-0.6,0,0.6).
\label{eq:oa-x-dgp}
\end{equation}
Untreated individual outcomes are generated as
\begin{equation}
Y_{it}(0)
=
A_{1t}(0)
+
\gamma_{Z_i}
+
0.4X_i
+
0.30X_i\ind\{t=1\}
+
\varepsilon_{it},
\label{eq:oa-standardization-dgp}
\end{equation}
and the post-treatment effect is
\begin{equation}
\tau(z,x)=\tau_z+0.15x,
\qquad
(\tau_1,\tau_2,\tau_3)=(0.35,0.25,0.15).
\label{eq:oa-standardization-effect}
\end{equation}
The marginal target distribution of $X$ has mean zero, so the standardized effects are 0.35, 0.25, and 0.15.

The observed pre--post change contains $0.30X$ from the differential untreated trend and $0.15X$ from variation in treatment effects. Since the marginal target distribution of $X$ has mean zero, the unadjusted framework therefore yields
\[
\tau_g+(0.30+0.15)\mu_g
=
\tau_g+0.45\mu_g,
\]
or 0.08, 0.25, and 0.42 across the three groups.

I use $T_0=20$ aggregate pretreatment periods and 2,000 observations in each repeated cross-section. I fit a conditional mean model with group-specific levels and changes, allowing both to vary linearly with $X$. For each group, I average the fitted change over the same empirical marginal distribution of $X$ and subtract the fitted change averaged over the empirical joint distribution of $(Z,X)$. The experiment uses 1,000 Monte Carlo replications; Table~\ref{tab:oa-standardization} reports the resulting subgroup-specific estimates.

\begin{table}[t]
\centering
\caption{Covariate-standardization experiment}
\label{tab:oa-standardization}
\begin{threeparttable}
\small
\setlength{\tabcolsep}{2pt}
\begin{tabular}{@{}l S[table-format=1.3] S[table-format=1.3] S[table-format=1.3] S[table-format=-1.3] S[table-format=1.3]@{}}
\toprule
Group & {\shortstack{True\\standardized\\effect}} & \multicolumn{2}{c}{Monte Carlo mean} & \multicolumn{2}{c}{Standardized estimator} \\
\cmidrule(lr){3-4}\cmidrule(l){5-6}
 & & {Unadjusted} & {Standardized} & {Bias} & {RMSE} \\
\midrule
$Z=1$ & 0.350 & 0.090 & 0.361 & 0.011 & 0.237 \\
$Z=2$ & 0.250 & 0.259 & 0.259 & 0.009 & 0.233 \\
$Z=3$ & 0.150 & 0.429 & 0.160 & 0.010 & 0.236 \\
\bottomrule
\end{tabular}
\begin{tablenotes}[flushleft]
\footnotesize
\item \textit{Notes:} Results are based on 1,000 Monte Carlo replications. The unadjusted estimator centers group-specific observed changes without standardizing the distribution of $X$. The standardized estimator averages the conditional change for each group over the common target distribution of $X$ before combining it with the population-average effect.
\end{tablenotes}
\end{threeparttable}
\end{table}

\section{Massachusetts application: data and estimation}
\label{sec:mass-data}

The application combines the Massachusetts Health Reform Survey (MHRS) and the Behavioral Risk Factor Surveillance System (BRFSS). I use the 2006 MHRS \citep{long2006mhrs} and the 2007 MHRS \citep{long2007mhrs} repeated cross-sections to estimate education-specific standardized treatment-effect contrasts within Massachusetts. Quarterly BRFSS data from 2000Q1 through 2010Q4 provide the state panel used to estimate the population-average effect \citep{cdc2013brfss}. Both analyses are restricted to adults aged 18--64. The resulting analytic MHRS samples contain \MhrsCountPre\ respondents in 2006 and \MhrsCountPost\ in 2007; the state panel contains Massachusetts and \ComparisonStateCount\ comparison states.

I code current insurance coverage as a binary outcome in both sources. In the MHRS, I reconstruct current coverage from the insurance-status sequence, including respondents whose coverage status is established before later items are skipped. The resulting survey-weighted coverage rates are 86.05 percent in 2006 and 93.51 percent in 2007. For BRFSS, I harmonize the current-health-plan item across panel years.

Massachusetts enacted its major health-insurance reform in 2006 and implemented it in stages through 2007. Some provisions were already in effect when the 2006 MHRS was fielded, so the 2006 survey represents a partial-implementation baseline. The main analysis dates the subsequent policy transition to 2007Q1 and evaluates the population-average effect from the comparative-case analysis in 2007Q4, the quarter most closely aligned with the 2007 MHRS wave. Under this treatment definition, provisions already operating before 2007Q1 form part of the baseline regime, and the aggregate counterfactual represents coverage under continuation of that regime after 2007Q1.

The heterogeneity analysis distinguishes adults with high school education or less, some college or an associate degree, and a bachelor's degree or higher. The target population is the population of Massachusetts adults aged 18--64 represented by the 2007 MHRS survey weights. Let $F_{X,Z}^\star$ denote its joint distribution of education $Z$ and adjustment variables $X$, where $X$ consists of age and sex. The standardized education-specific effects use the marginal age-sex distribution $F_X^\star$ for every education group,\footnote{The standardization uses the full 2007 MHRS age-sex distribution. Where this distribution extends beyond the empirical age support of an education group, the application relies on extrapolation from the fitted conditional mean model. Applying Proposition~\ref{prop:standardized-effects} to a target population restricted to common support would require a compatible population-average effect for that restricted population.} while the population-average observed change is evaluated over the joint distribution $F_{X,Z}^\star$.

I estimate conditional insurance coverage using a survey-weighted linear working model with education- and period-specific intercepts, age profiles, and sex coefficients. Let $P_i$ denote the post-treatment indicator, $B(A_i)$ a natural cubic-spline basis in age with three degrees of freedom, and $S_i$ an indicator for female sex. Using $Z_i$ for education group, the conditional mean model is
\begin{equation}
\E[Y_i\mid Z_i=z,P_i=p,A_i=a,S_i=s]
=
\alpha_{zp}
+
B(a)^{\mathsf T}\beta_{zp}
+
\gamma_{zp}s.
\label{eq:oa-micro-model}
\end{equation}
This parameterization allows the pre--post change to vary jointly with education, age, and sex.

Let $\widehat\mu_p(z,x)$ denote the fitted conditional mean in period $p$. The standardized observed change for education group $z$ is
\begin{equation}
\widehat\Delta^S(z)
=
\int
\left\{
\widehat\mu_1(z,x)-\widehat\mu_0(z,x)
\right\}
\,d\widehat F_X^\star(x),
\label{eq:oa-standardized-change}
\end{equation}
whereas the target-population observed change is
\begin{equation}
\widehat{\bar\Delta}
=
\int
\left\{
\widehat\mu_1(z,x)-\widehat\mu_0(z,x)
\right\}
\,d\widehat F_{X,Z}^\star(x,z).
\label{eq:oa-target-change}
\end{equation}
Their difference estimates the standardized treatment-effect contrast between group $z$ and the population-average effect. Because the first quantity evaluates every education group at the same marginal age-sex distribution while the second averages over the observed joint distribution, these standardized treatment-effect contrasts need not average to zero across education groups.

Massachusetts is the treated state. Maine and Vermont are excluded because both introduced statewide coverage reforms during the 2000--2010 analysis period \citep{maineDirigo2003,vermontCatamount2006}; the District of Columbia and U.S. territories are outside the state-level donor pool. I construct quarterly state-level coverage rates using the BRFSS survey weights and estimate the Massachusetts counterfactual using the generalized synthetic control method \citep{xu2017generalized}, implemented in \texttt{fect} \citep{liu2024practical}. I compare ranks zero through five by rolling cross-validation over untreated observations and select the rank with the lowest prediction error. Rank zero minimizes the criterion, so the fitted model contains additive state and period effects but no interactive factors.

I use 1,999 bootstrap draws for each component. For the population-average effect, I use the parametric bootstrap under the fitted comparative-case model and center the resulting errors before adding them to the point estimate. For the within-jurisdiction component, primary sampling units are resampled with replacement within wave-specific survey strata, with singleton strata retained. The supplied survey weights are held fixed at their observed values and multiplied by the resampled primary-sampling-unit frequencies within each resample. Each resample refits the conditional mean model and re-estimates the 2007 target distributions and standardized treatment-effect contrasts. The aggregate BRFSS panel and MHRS are based on separate survey samples, so I combine independently drawn bootstrap errors from the two components.

To construct the direct comparison shown in Figure~2(c), I use BRFSS microdata to form separate education-specific state-quarter panels standardized to the same 2007 MHRS age-sex target distribution and estimate each effect using the same treatment timing and comparative-case estimator. Their confidence intervals condition on the estimated 2007 MHRS target distribution.

\end{document}